\documentclass[crcready]{wias2e}

\usepackage[T1]{fontenc}
\usepackage{times}%

\usepackage[numbers]{natbib}
\usepackage{amsmath}
\usepackage{endnotes}

\usepackage{amsthm}
\usepackage{epsfig}

\usepackage{amssymb} 
\usepackage{amsmath}
\usepackage{amsthm}

\usepackage{algorithm}
\usepackage{algpseudocode}

\usepackage{multicol}

\usepackage{tablefootnote}

\newtheorem{remark}{Remark}
\newtheorem{example}{Example}

\usepackage{url}

\usepackage{tikz}
\usetikzlibrary{plotmarks}
\usepackage{pgfplots}
\pgfplotsset{plot coordinates/math parser=false}
\newlength\figureheight
\newlength\figurewidth

\usepackage{adjustbox}
\usepackage{array}
\newcolumntype{R}[2]{%
    >{\adjustbox{angle=#1,lap=\width-(#2)}\bgroup}%
    l%
    <{\egroup}%
}
\usepackage{multirow} 
\usepackage{booktabs}
\usepackage{relsize}
\newcommand{\ra}[1]{\renewcommand{\arraystretch}{#1}}

\newcolumntype{d}[1]{D{.}{.}{#1}}

\def\VR{\kern-\arraycolsep\strut\vrule &\kern-\arraycolsep}
\def\VVR{\kern-\arraycolsep\strut\vrule\hspace{0.05em}\vrule &\kern-\arraycolsep}
\def\vr{\kern-\arraycolsep & \kern-\arraycolsep}

\firstpage{1}
\lastpage{5}
\volume{5}
\pubyear{2015}

\begin{document}
\begin{frontmatter}                           

%
\title{Top-N recommendations in the presence of sparsity: An NCD-based approach}

\runningtitle{Top-N recommendations in the presence of sparsity: An NCD-based approach}

\author[A,B]{\fnms{Athanasios N.} \snm{Nikolakopoulos}\thanks{Corresponding author. E-mail: nikolako@ceid.upatras.gr. Tel.no: +302610997543}} and
\author[A,B]{\fnms{John D.} \snm{Garofalakis}\thanks{ E-mail: garofala@ceid.upatras.gr. Tel.no: +302610997562}}
\runningauthor{Nikolakopoulos et al.}
\address[A]{Department of Computer Engineering and Informatics, University of Patras, Panepistimioupoli, GR26500, Rio,\\ Greece}
\address[B]{Computer Technology Institute and Press ``Diophantus'', Panepistimioupoli, GR26504, Rio, Greece\\
E-mail: \{nikolako,garofala\}@ceid.upatras.gr}

\begin{abstract}
Making recommendations in the presence of sparsity is known to present one of the most challenging problems faced by collaborative filtering methods. In this work we tackle this problem by exploiting the innately hierarchical structure of the item space following an approach inspired by the theory of \textit{Decomposability}.  We view the itemspace as a Nearly Decomposable system and we define blocks of closely related elements and corresponding indirect proximity components. We study the theoretical properties of the decomposition and we derive sufficient conditions that guarantee full item space \textit{coverage} even in \textit{cold-start} recommendation scenarios. A comprehensive set of experiments on the MovieLens and the Yahoo!R2Music datasets, using several widely applied performance metrics, support our model's theoretically predicted properties and verify that NCDREC outperforms several state-of-the-art algorithms, in terms of recommendation accuracy, diversity and sparseness insensitivity.

\end{abstract}


\end{frontmatter}

\section{Introduction}

Recommender Systems (RS) are information filtering tools that have been widely adopted over the past decade, by the majority of e-commerce sites, in order to make intelligent personalized product suggestions to their customers~\cite{Tivo,tapestry,amazon}.   RS technology enhances user experience and it is known to increase user fidelity to the system~\cite{ricci2011introduction}. Correspondingly, from an economic perspective, the utilization of recommender systems is known to assist in building bigger, and more loyal customer bases, and to drive a significant increase in the volume of product sales~\cite{Sales1,Sales2,Schafer:1999:RSE:336992.337035}. 

The development of recommender systems is -- in a very fundamental sense -- based on a rather simple observation: people, very often rely their every day decision making on advise and suggestions provided by the community. For example, it is very common when one wants to pick a new movie to watch, to take into consideration published reviews about the movie or ask friends for their opinion. Mimicking this behavior, recommender systems exploit the plethora of information produced by the interactions of a large community of users, and try to deliver personalized suggestions that aim to help an active user cope with the devastating number of options in front of him.

\begin{figure*}[tph!]
\centering
\epsfig{file=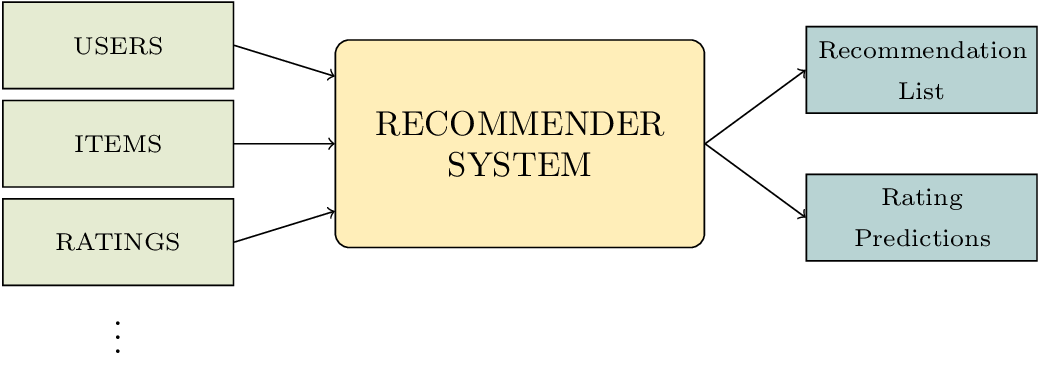,width=0.75\linewidth}
\caption[Example Recommender System]{Example Recommender System}
\label{fig:system-combination-eng}
\end{figure*}

Among the several different approaches to building recommender systems, Collaborative Filtering (CF) is widely regarded as one of the most successful ones~\cite{Tivo,Herlocker:1999:AFP:312624.312682,amazon,GroupLens,Sarwar:2001:ICF:371920.372071}. CF methods basically model both users and items as sets of ratings, and focus on the sparse rating matrix that lies at the common core, trying to either estimate the missing values, or find promising cells to propose (see Figure~\ref{fig:system-combination-eng}). In the majority of CF related work for reasons of mathematical convenience (as well as fitness with formal optimization methods), the recommendation task reduces to predicting the ratings for all the unseen user-item pairs (\textit{prediction-based} methods~\cite{desrosiers2011comprehensive,koren2011advances,Vozalis:2006:ECF:1239776.1239777}).  Recently, however, many leading researchers have turned significant attention to \textit{ranking-based} methods which are believed to conform more naturally with how the recommender system will actually be used in practice~\cite{Balakrishnan:2012:CR:2124295.2124314,Cremonesi:2010:PRA:1864708.1864721,Fouss:2007:RCS:1263132.1263335,Freno:2009:SPE:1557019.1557059,Gori:2007:IRB:1625275.1625720,marlin2009collaborative,Lanczos,COSRANK_ACOSRANK:2008}.

Despite their success in many application settings, RS techniques suffer a number of problems that remain to be resolved.
One of the most important such problems arises from the fact that often available data are insufficient for identifying similar elements and is commonly referred to as the \textit{Sparsity Problem}. Sparsity imposes serious limitations to the quality of recommendations, and it is known to decrease significantly the diversity and the effectiveness of CF methods -- especially in recommending unpopular items (\textit{``long tail''} problem). Unfortunately, sparsity is an intrinsic characteristic of recommender systems because in the majority of realistic applications, users typically interact with only a small portion of the available items, and the problem is aggravated even more, by the fact that new users with no ratings at all, are regularly added to the system (\textit{Cold-Start} problem~\cite{Bobadilla:2013:RSS:2483330.2483573,NikolakopoulosK15}).

Among the most promising approaches in dealing with limited coverage and sparsity are \textit{graph-based} methods \cite{desrosiers2011comprehensive,Fouss:2012:EIK:2207289.2207578,Gori:2007:IRB:1625275.1625720,COSRANK_ACOSRANK:2008}. The methods of this family exploit transitive relations in the data, which makes them able to estimate the relationship between users and items that are not directly connected.   Gori and Pucci \cite{Gori:2007:IRB:1625275.1625720} proposed ItemRank; a PageRank-inspired scoring algorithm that produces a personalized ranking vector using a random walk with restarts on an items' correlation graph induced by the ratings.  
Fouss et al.  \cite{Fouss:2012:EIK:2207289.2207578,Fouss:2007:RCS:1263132.1263335} create a graph model of the RS database and they present a number of methods to compute node similarity measures, including the random walk-related average Commute Time and average First Passage Time, as well as the pseudo-inverse of the graph's Laplacian. They compare their methods against other state-of-the-art graph-based approaches such as, the sophisticated node similarity measure that integrates indirect paths in the graph, based on the matrix-forest theorem \cite{chebotarev2006matrix}, and a similarity measure based on the well known Katz algorithm \cite{KATZ}. 

Here, we attack the sparsity problem from a different perspective. The fact, that sparsity has been commonly observed in models of seemingly unrelated naturally emerging systems, suggests an even more fundamental cause behind this phenomenon. According to Herbert A. Simon, this inherent sparsity is intertwined with the structural organization and the evolutionary viability of these systems. In his seminal work on the architecture of complexity~\cite{Simon:1996:SA:237774}, he argued that the majority of sparse hierarchically structured systems share the property of having a \textit{Nearly Completely Decomposable} (NCD) architecture: they can be seen as comprised of a hierarchy of interconnected blocks, sub-blocks and so on, in such a way that elements within any particular such block relate much more vigorously with each other than do elements belonging to different blocks, and this property holds between any two levels of the hierarchy. 

The analysis of decomposable systems has been pioneered by Simon and Ando~\cite{simon1961aggregation} who reported on state aggregation in linear models of economic systems, but the universality and the versatility of Simon's idea have permitted the theory to be used in many complex problems from diverse disciplines ranging from economics, cognitive theory and social sciences, to computer systems performance evaluation, data mining and information retrieval \cite{cevahir2011site,courtois1977decomposability,blockrank,meyer2012stochastic,MeyerRV13,nikolakopoulos2013ncdawarerank,yin2013continuous}.

The criteria behind the decomposition vary with the goals of the study and the nature of the problem under consideration. For example, in the stochastic modeling literature, decomposability is usually found in the \textit{time domain} and the blocks are defined to separate the short-term from the long-term temporal dynamics~\cite{courtois1977decomposability,yin2013continuous}. In other cases the decomposition is chosen to highlight known \textit{structural  properties} of the underlying space; for example in the field of link analysis, many leading researchers have exploited the nearly decomposable structure of the Web, 
from a computational (faster extraction of the PageRank vector) as well as a qualitative (generalization of the random surfer teleportation model) perspective~\cite{cevahir2011site,blockrank,nikolakopoulos2013ncdawarerank}.

In this work\footnote{This work extents significantly our initial contribution~\cite{NCDRECConferencePaper}, adding detailed  presentation of the NCDREC model enriched by thorough explanations and examples, as well as rigorous theoretical analysis of its constituents parts. Furthermore, in this paper we provide a more in-depth coverage of related literature including thorough discussions of the competing state-of-the-art recommendation techniques as well as details regarding their implementation in our experiments.}, building on the intuition behind NCD, we decompose the item space into \textit{blocks}, and we use these blocks to characterize the inter-item \textit{proximity} in a macroscopic level. Central to our approach is the idea that blending together the direct with the indirect inter-item relations can help reduce the sensitivity to sparseness and improve the quality of recommendations. To this end, we propose \textbf{NCDREC}, a novel ranking based recommendation method which:
\begin{itemize}
\item Provides a theoretical framework that enables the exploitation of item space's innately decomposable structure in an efficient, and scalable way. 
\item Produces recommendations that outperform several state-of-the-art methods, in widely used metrics (Section~\ref{SubSec:StandardRecommendation}), achieving high quality results even in the  generally harder task of recommending \textit{long-tail} items (Section~\ref{SubSec:LongTailTests}).   
\item Displays low sensitivity to the problems caused by the sparsity of the underlying space and treats \textit{New Users} more fairly;  this is supported both by NCDREC's theoretical properties (Section~\ref{ColdStartSubcomponent}) and our experimental findings (Section~\ref{SubSec:NewUsersTests}).  
\end{itemize}

The rest of the paper is organized as follows. 
In Section~\ref{NCDREC_Framework}, after discussing briefly the intuition behind the exploitation of Decomposability for recommendations, we introduce formally our model and we study several of its interesting theoretical properties (Section~\ref{SubSec:NCDREC_Model}). In Section~\ref{SubSec:NCDREC_Algorithm}  we present the NCDREC algorithm and we talk about its storage and computational aspects. Our testing methodology and experimental results are presented in Section~\ref{Sec_Experiments}. Finally, Section~\ref{Sec_Conclusion} concludes this paper and outlines directions for future work. 

\section{NCDREC Framework}
\label{NCDREC_Framework}

\subsection{Exploiting Decomposability for Recommendations}

In the method we propose in this work, we see the set of items as a decomposable space and, following the modeling approach of a recently proposed Web ranking framework~\cite{nikolakopoulos2013ncdawarerank,RandomSurfingWithoutTeleportation}, we use the decomposition to characterize  macro-relations between the elements of the dataset that can hopefully refine and augment the underlying collaborative filtering approach and ``fill in'' some of the void left by the intrinsic sparsity of the data. The criteria behind the decomposition can vary with the particular aspects of the item space, the information available etc. For example, if one wants to recommend hotels, the blocks may be defined to depict geographic information; in the movie recommendation problem, the blocks may correspond to the categorization of movies into genres, or other movie attributes etc. To give our framework maximum flexibility, we extend the notion to allow overlapping blocks; intuitively this seems to be particularly useful in many modeling  approaches and recommendation problems.

Before we proceed to the rigorous definition of the NCDREC framework, we outline briefly our approach: First, we define a decomposition, $\mathcal{D}$, of the item space into blocks and we introduce the notion of $\mathcal{D}$-proximity, to characterize the implicit inter-level relations between the items. Then, we translate this proximity notion to suitably defined matrices that quantify these macroscopic inter-item relations under the prism of the chosen decomposition. These matrices need to be easily handleable in order for our method to be applicable in realistic scenarios. Furthermore, their contribution to the final model needs to be weighted carefully so as not to ``overshadow'' the pure collaborative filtering parts of the model. In achieving these, we follow an approach based on perturbing the standard CF parts, using suitably defined low-rank matrices. Finally, to fight the inevitably extreme and localized sparsity related to cold start scenarios we create a Markov chain-based subcomponent, designed to increase the percentage of the item space covered by the produced recommendations, and we study the conditions (in terms of theoretical properties of the proposed decomposition) under which full item space coverage is guaranteed.  

\subsection{NCDREC Model and Theoretical Properties}
\label{SubSec:NCDREC_Model}
\newtheorem{proposition}{Proposition}
\newtheorem{theorem}{Theorem}
\newtheorem{lemma}{Lemma}
\newtheorem{definition}{Definition}

\subsubsection{Notation}
All vectors are represented by bold lower case letters and they are column vectors (e.g., $\boldsymbol{\omega}$). All matrices are represented by bold upper case letters (e.g., $\mathbf{W} $). The $i^{\text{th}}$ row and $j^{\text{th}}$ column of matrix $\mathbf{W}$ are denoted $\mathbf{w}^\intercal_{i}$ and $\mathbf{w}_{j}$, respectively. The $ij^{th}$ element of matrix $\mathbf{W}$ is denoted $[\mathbf{W}]_{ij}$. We use $\operatorname{diag}(\boldsymbol{\omega})$ to denote the matrix having vector $\boldsymbol{\omega}$ on its diagonal, and zeros elsewhere. We use calligraphic letters to denote sets (e.g., $\mathcal{U,V}$). Finally, symbol $\triangleq$ is used in definition statements.

\subsubsection{Definitions}
\label{ModelDefinitions}
Let $\mathcal{U} = \{u_1,\dots,u_n\}$ be a set of \textit{users}, $\mathcal{V} = \{v_1,\dots,v_m\}$ a set of \textit{items} and $\mathcal{R}$ a set of tuples 
\begin{equation}
\mathcal{R} \triangleq \{t_{ij}\} = \{(u_i,v_j,r_{ij})\},
\end{equation} where $r_{ij}$ is a nonnegative number referred to as the \textit{rating} given by user $u_i$ to the item $v_j$. For each user in $\mathcal{U}$ we assume he has rated at least one item; similarly each item in $\mathcal{V}$ is assumed to have been rated by at least one user. 

We define an associated user-item \textbf{rating matrix} $\mathbf{R}\in\mathbb{R}^{n \times m}$, whose $ij^{th}$ element equals $r_{ij}$, if $t_{ij}\in\mathcal{R}$, and zero otherwise. 
For each user $u_i$, we denote $\mathcal{R}_{i}$  the set of items rated by $u_i$ in $\mathcal{R}$, and we define a \textbf{preference vector} $\boldsymbol{\omega}\triangleq [\omega_1,\dots,\omega_m]$, whose nonzero elements contain the user's ratings that are included in $\mathcal{R}_{i}$, normalized to sum to one. 

We consider an indexed family of non-empty sets 
\begin{equation}
		\mathcal{D} \triangleq \{\mathcal{D}_1,\dots,\mathcal{D}_K\},
\end{equation}
that defines a $\mathcal{D}$-\textbf{decomposition} of the underlying space $\mathcal{V}$, such that $\mathcal{V}=\bigcup_{k=1}^{K}\mathcal{D}_k$. Each set $\mathcal{D}_I$ is referred to as a $\mathcal{D}$-\textbf{block}, and its elements are considered related according to some criterion.

We define 
\begin{equation}
\mathfrak{D}_v\triangleq \bigcup_{{v \in \mathcal{D}_k}}\mathcal{D}_k
\end{equation} 
to be the \textbf{proximal set} of items of $v \in \mathcal{V}$, i.e. the union of the $\mathcal{D}$-blocks that contain $v$. We use $N_v$ to denote the number of different blocks in $\mathfrak{D}_v$, and 
\begin{equation}
n^\ell_{u_i} \triangleq \lvert{\{r_{ik}: (r_{ik}>0) \wedge (v_k\in \mathcal{D}_\ell)\}}\rvert
\end{equation}
 for the number of items rated by user $u_i$ that belong to the $\mathcal{D}$-block, $\mathcal{D}_\ell$. 
Every $\mathcal{D}$-decomposition is also associated with an undirected graph
\begin{equation}
\mathcal{G}_{\mathcal{D}}\triangleq(\mathcal{V}_{\mathcal{D}},\mathcal{E}_{\mathcal{D}})
\end{equation} 
Its vertices correspond to the $\mathcal{D}$-blocks, and an edge between two vertices exists whenever the intersection of these blocks is a non-empty set. 
This graph is referred to as the \textbf{block coupling graph} for the $\mathcal{D}$-decomposition. 

Finally, with every $\mathcal{D}$-decomposition we associate an \textbf{Aggregation matrix} $\mathbf{A}_{\mathcal{D}}\in \mathbb{R}^{m \times K}$, whose $jk^{th}$ element is 1, if $v_j \in \mathcal{D}_k$ and zero otherwise.

\subsubsection{Main Component}
The pursuit of ranking-based recommendations, grants us the flexibility of not caring about the exact recommendation scores; only the correct item ordering is needed. This allows us to manipulate the missing values of the rating matrix in an ``informed'' way so as to introduce some preliminary ordering based on the user's expressed opinions about some items, and the way these items relate with the rest of the item space.  

The existence of such connections is rooted in the idea that a user's rating, except for expressing his direct opinion about a particular item, also gives a clue about his opinion regarding the proximal set of this item. So, ``propagating''  these opinions through the decomposition to the many related elements of the item space, can hopefully refine the estimation of his preferences regarding the vast fraction of the item set for which he has not expressed opinions, and introduce an ordering between the zeros in the rating matrix, that will hopefully  relieve sparsity related problems.  

Having this in mind, we perturb the user-item rating matrix $\mathbf{R}$, with an \textbf{NCD preferences matrix} $\mathbf{W}$ that propagates the expressed user opinions about particular items to the proximal sets. The resulting matrix is given by: 
\begin{equation}
\mathbf{G} \triangleq \mathbf{R} + \epsilon\mathbf{W},
\end{equation}
where $\epsilon$ is a positive parameter, chosen small so as not to ``eclipse''  the actual ratings. The NCD preferences matrix is formally defined below:

\noindent\textbf{NCD Preferences Matrix $\mathbf{W}$}. The NCD preferences matrix, is defined to propagate each user's ratings to the many related elements (in the $\mathcal{D}$-decomposition sense) of the item space. Formally, matrix $\mathbf{W}$ is defined as follows: 
\begin{equation}
\mathbf{W}  \triangleq  \mathbf{Z} \mathbf{X}^{\intercal}
\end{equation}
where matrix $\mathbf{X}$ denotes the row normalized version of $\mathbf{A}_{\mathcal{D}}$, and the $ik^{th}$ element of matrix $\mathbf{Z}$ equals $(n^k_{u_i})^{-1}[\mathbf{R}\mathbf{A}_{\mathcal{D}}]_{ik}$, when $n^k_{u_i}>0$, and zero otherwise.

The final recommendation vectors are produced by projecting the perturbed data onto an $f$-dimensional space. In particular, the final recommendation vectors are defined to be the rows of matrix 
\begin{equation}
\mathbf{\Pi} \triangleq \mathbf{U_f \Sigma_f V_f^\intercal},
\label{Eq:MainComponentRecVectors}
\end{equation} 
 where matrix $\mathbf{\Sigma_f}\in \mathbb{R}^{f\times f}$
is a diagonal matrix containing the first $f$ singular values of $\mathbf{G}$, and matrices $\mathbf{U_f}\in \mathbb{R}^{n\times f}$ and $\mathbf{V_f}\in\mathbb{R}^{m\times f}$ are orthonormal matrices containing the corresponding left and right singular vectors.

\begin{remark}
	{\upshape 
	In fact, the recommendation vectors produced by  Eq.~(\ref{Eq:MainComponentRecVectors}) can be seen as arising from a low dimensional eigenspace of an NCDaware inter-item similarity matrix. We discuss this further in Appendix \ref{Ap:MainComponentDiscussion}.
}
\end{remark}

\subsubsection{ColdStart Component}
\label{ColdStartSubcomponent}
In some cases the sparsity phenomenon becomes so intense and localized that the perturbation of the ratings through matrix $\mathbf{W}$ is not enough. Take for example newly emerging users in an existing recommender system. 
Naturally, because these users are new, the number of ratings they introduce in the RS is usually not sufficient to be able to make reliable recommendations. If one takes into account only their direct interactions with the items, the recommendations to these newly added users are very likely to be restricted in small subsets of $\mathcal{V}$, leaving the majority of the item space uncovered.  

\begin{figure*}[!htpb]
	\centering
	\epsfig{file=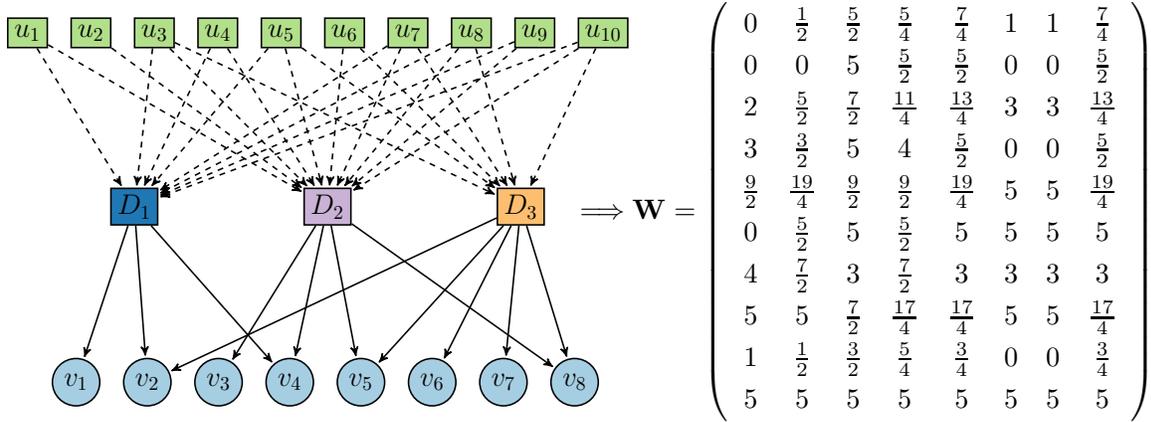,width=0.95\linewidth}
	\caption{We see matrix $\mathbf{W}$ that corresponds to the Example \ref{Example}.  }
	\label{fig:Example1}
\end{figure*}

To address this problem which represents one of the continuing difficulties faced by recommender systems in operation~\cite{Bobadilla:2013:RSS:2483330.2483573}, we create a {\sc ColdStart} subcomponent based on a discrete Markov chain model over the item space with transition probability matrix $\mathbf{S}$, defined to bring together the direct as well as the decomposable structure of the underlying space. Matrix $\mathbf{S}$ is defined to consist of three components, namely  a \textbf{rank-one preference matrix} $\mathbf{e}\boldsymbol{\omega}^{\intercal}$ that rises from the explicit ratings of the user as presented in the training set; a \textbf{direct proximity matrix} $\mathbf{H}$, that depicts the direct inter-item relations; and an \textbf{NCD proximity matrix} $\mathbf{D}$ that relates every item with its proximal sets. Concretely, matrix $\mathbf{S}$ is given by:
\begin{equation}
\mathbf{S} \triangleq (1-\alpha) \mathbf{e}\boldsymbol{\omega}^{\intercal} + \alpha(\beta \mathbf{H} + (1-\beta) \mathbf{D}) 
\end{equation}
with $\alpha$ and $\beta$ being positive real numbers for which $\alpha,\beta < 1$ holds. Parameter $\alpha$ controls how frequently the Markov chain ``restarts'' to the preference vector, $\boldsymbol{\omega}$, whereas parameter $\beta$ weights the involvement of the Direct and the NCD Proximity matrices in the final Markov chain model. The personalized ranking vector for each newly added user is defined to be the \textit{stationary probability distribution} of the Markov chain that corresponds to the stochastic matrix $\mathbf{S}$, using the normalized ratings of the user as the initial distribution. 
\begin{description}
	\item [\textbf{Direct Proximity Matrix $\mathbf{H}$}.]
	The direct proximity matrix $\mathbf{H}$ is designed to capture the direct relations between the elements of $\mathcal{V}$. Generally, every such element will be associated with a discrete distribution $\mathbf{h}_v = [h_1,h_2,\cdots,h_m]$ over $\mathcal{V}$, that reflects the correlation between these elements. In our case, we use the stochastic matrix defined as follows:
	\begin{equation}
	\mathbf{H} \triangleq \operatorname{diag}(\mathbf{Ce})^{-1}\mathbf{C}
	\end{equation}
	where $\mathbf{C}$ is an $m\times m$ matrix whose $ij^{th}$ element is defined to be $[\mathbf{C}]_{ij} \triangleq \mathbf{r^{\intercal}_i}\mathbf{r_j}$ for $i \neq j$, zero otherwise, and $\mathbf{e}$ is a properly sized unit vector. 
	\item [\textbf{NCD Proximity Matrix $\mathbf{D}$}.]
	The NCD proximity matrix $\mathbf{D}$ is created to depict the interlevel connections between the elements of the item space.
	In particular, each row of matrix $\mathbf{D}$ denotes a probability vector $\mathbf{d}_v$, that distributes evenly its mass between the $N_v$ blocks of $\mathfrak{D}_v$, and then, uniformly to the included items of each block. Formally, matrix $\mathbf{D}$ is defined by: 
	\begin{equation}
	\mathbf{D}\triangleq \mathbf{X}\mathbf{Y}
	\label{def:D}
	\end{equation}
	where $\mathbf{X,Y}$ denote the row normalized versions of $\mathbf{A}_{\mathcal{D}}$ and $\mathbf{A}^{\intercal}_{\mathcal{D}}$ respectively. 
\end{description}

\begin{lemma}
	Matrices $\mathbf{H,D}$ are well defined row stochastic matrices.
\end{lemma}
\begin{proof}
	We will begin with matrix $\mathbf{H}$. First, notice that for matrix $\mathbf{H}$ to be well defined it is necessary $\operatorname{diag}(\mathbf{Ce})$ to be invertible. But this is assured by our model's assumption that every item have been rated by at least one user. Indeed, when this assumption holds, every row of matrix $\mathbf{C}$ denotes a non-zero vector in $\mathfrak{R}^m$, thus $\mathbf{C}\mathbf{e}$ denotes a vector of strictly positive elements, which makes the diagonal matrix $\operatorname{diag}(\mathbf{Ce})$ invertible, as needed.
	
	For matrix $\mathbf{D}$ it suffices to show that for any $\mathcal{D}$-decomposition, every column and every row of the corresponding aggregation matrix $\mathbf{A}_\mathcal{D}$, denote non-zero vectors in $\mathfrak{R}^m$ and $\mathfrak{R}^K$ respectively.
	The latter is ensured from the fact that NCD blocks are defined to be non-empty, whereas the former condition holds because the union of the $\mathcal{D}$-blocks denote a cover of the itemspace. 
\end{proof}

\begin{example}
		\label{Example}	
	{\upshape 
		
	To clarify the definition of the NCD matrices $\mathbf{W,D}$, we give the following example. Consider a simple movie recommendation system consisting of an itemspace of 8 movies and a userspace of 10 users each having rated at least one movie. Let the set of ratings, $\mathcal{R}$, be the one presented below:

	\begin{equation}
	 \mathcal{R}=\left\{\begin{matrix}
	 (u_4,v_1,1),& (u_7,v_1,4),&(u_8,v_1,5),\\
	 (u_{10}, v_1, 5),& (u_5, v_2, 5),& (u_1, v_3, 4),\\
	 (u_2, v_3, 5),&(u_8, v_3, 2),&(u_9, v_3, 2),\\
	 (u_{10}, v_3, 5),&(u_3, v_4, 2),&(u_4, v_4, 5),\\
	 (u_5, v_4, 4),&(u_9, v_4, 1),&(u_1, v_5, 1),\\
	 (u_5, v_5, 5),&(u_6, v_5, 5),&(u_7, v_5, 3),\\
	 (u_3, v_6, 3),&(u_{10}, v_6, 5),&(u_3, v_7, 1),\\
	 (u_3, v_8, 5),&(u_6, v_8, 5),&(u_8, v_8, 5),
	  \end{matrix}\right\}
	\end{equation}

\begin{figure*}[!htpb]
	\centering
	\epsfig{file=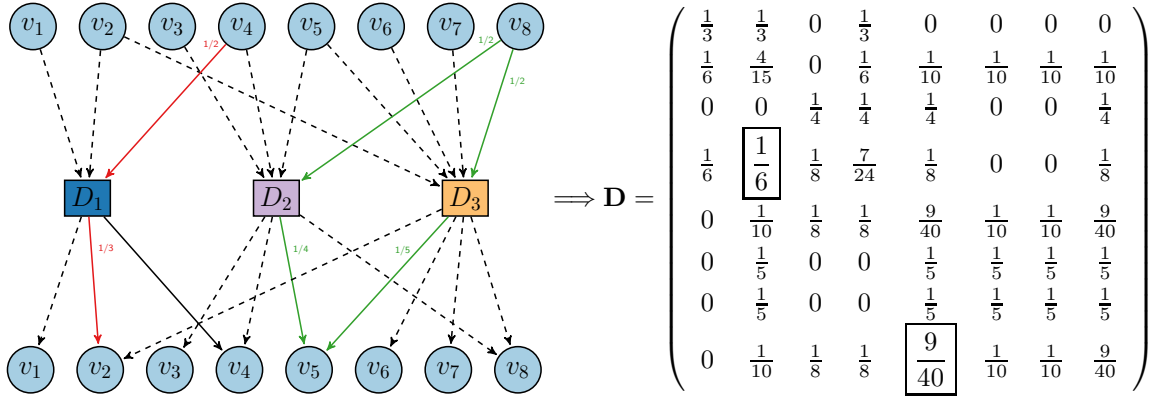,width=0.95\linewidth}
	\caption{We see the matrix $\mathbf{D}$ that corresponds to Example \ref{Example}. 
		We highlight with red and green color the computation of the $[\mathbf{D}]_{42}$ and $[\mathbf{D}]_{85}$, respectively. }
	\label{fig:Example2}
\end{figure*}

	 Assume also that the 8 movies of the itemspace belong to 3 genres as seen below:  
	 
 		\begin{equation}
 		\bordermatrix{~ & \mathcal{D}_1 & \mathcal{D}_2 & \mathcal{D}_3 & \vr \mathcal{N}_v  \cr
 			v_1 & \checkmark & - & - & \VVR 1   \cr
 			v_2 & \checkmark & - & \checkmark & \VVR 2  \cr
 			v_3 & - & \checkmark & - & \VVR 1 \cr
 			v_4 & \checkmark & \checkmark & - & \VVR 2 \cr
 			v_5 & - & \checkmark & \checkmark & \VVR 2 \cr
 			v_6 & - & - & \checkmark & \VVR 1 \cr
 			v_7 & - & - & \checkmark & \VVR 1 \cr
 			v_8 & - & \checkmark & \checkmark & \VVR 2 \cr}
 		\end{equation}
 		
 		The corresponding aggregation matrix $\mathbf{A}_{\mathcal{D}}\in \mathbb{R}^{8 \times 3}$ is
 		
 		\begin{equation}
	 	\mathbf{A}_{\mathcal{D}} = \left(\begin{array}{ccc}
	 	1 & 0 & 0\\
	 	1 & 0 & 1\\
	 	0 & 1 & 0\\
	 	1 & 1 & 0\\
	 	0 & 1 & 1\\
	 	0 & 0 & 1\\
	 	0 & 0 & 1\\
	 	0 & 1 & 1\\
	 	 \end{array}\right)
 		\end{equation}
	 
		Following the definition of matrix $\mathbf{W}$ we get the matrix shown in Figure \ref{fig:Example1}. For the factor matrices $\mathbf{Z,X}$ we have: 
		
		\begin{displaymath}
		\mathbf{Z} = \left(\begin{array}{ccc} 0 & \frac{5}{2} & 1\\ 0 & 5 & 0\\ 2 & \frac{7}{2} & 3\\ 3 & 5 & 0\\ \frac{9}{2} & \frac{9}{2} & 5\\ 0 & 5 & 5\\ 4 & 3 & 3\\ 5 & \frac{7}{2} & 5\\ 1 & \frac{3}{2} & 0\\ 5 & 5 & 5 \end{array}\right), \quad 		\mathbf{X} = \left(\begin{array}{ccc} 1 & 0 & 0\\ \frac{1}{2} & 0 & \frac{1}{2}\\ 0 & 1 & 0\\ \frac{1}{2} & \frac{1}{2} & 0\\ 0 & \frac{1}{2} & \frac{1}{2}\\ 0 & 0 & 1\\ 0 & 0 & 1\\ 0 & \frac{1}{2} & \frac{1}{2} \end{array}\right)
		\end{displaymath}

Similarly, in Figure~\ref{fig:Example2} we give the detailed computation of the inter-item NCD Proximity matrix $\mathbf{D}$ of the {\sc ColdStart} component.    
}\end{example}

\subsubsection{Theoretical Properties of the ColdStart Subcomponent}

Informally, the introduction of the NCD proximity matrix $\mathbf{D}$, helps the item space become more ``connected'', allowing the recommender to reach more items even for the set of newly added users. When the blocks are overlapping this effect becomes stronger, and in fact, item space coverage can be guaranteed under certain conditions. 

\begin{theorem}[ItemSpace Coverage]
If the block coupling graph $\mathcal{G}_{\mathcal{D}}$ is connected, there exists a unique steady state distribution $\boldsymbol{\pi}$ of the Markov chain corresponding to matrix $\mathbf{S}$ that depends on the preference vector $\boldsymbol{\omega}$; however, irrespectively of any particular such vector, the support of this distribution includes every item of the underlying space. 
\label{Theorem1}
\end{theorem}

\begin{proof} 
	Before we proceed to the actual proof, we will give a small sketch: 
 When $\mathcal{G}_{\mathcal{D}}$ is connected, the Markov chain induced by the stochastic matrix $\mathbf{S}$ consists of a single irreducible and aperiodic closed set of states, that includes all the items. To prove the irreducibility part, we will show that the NCD proximity stochastic matrix, that corresponds to a connected block coupling graph, ensures that starting from any particular state of the chain, there is a positive probability of reaching every other state. For the aperiodicity part we will show that matrix $\mathbf{D}$ makes it possible, for the Markov chain to return to any given state in consecutive time epochs. The above is true for every stochastic vector $\boldsymbol{\omega}$, and for every positive real numbers $\alpha,\beta<1$. 
  
 \begin{lemma}
 	The connectivity of $\mathcal{G}_{\mathcal{D}}$ implies the irreducibility of the Markov chain with transition probability matrix $\mathbf{D}$.
 \end{lemma}
 
 \begin{proof}
 	 From the decomposition theorem of Markov chains we know that the state space $\mathcal{S}$ can be partitioned uniquely as 
 	 \begin{equation}
 	 \mathcal{S} = \mathcal{T}\cup\mathcal{C}_1\cup\mathcal{C}_2\cup\cdots
 	 \end{equation}
 	 where $\mathcal{T}$ is the set of transient states, and the $\mathcal{C}_i$ are irreducible closed sets of persistent states \cite{grimmett2001probability}.
 	 
 	 Furthermore, since $\mathcal{S}$ is finite at least one state is persistent and all persistent states are non-null (see \cite{grimmett2001probability}, page 225). We will prove that the connectivity of $\mathcal{G}_{\mathcal{D}}$ alone, ensures that starting from this state $i$, we can visit every other state of the Markov chain. In other words, the connectivity of $\mathcal{G}_{\mathcal{D}}$ implies that $\mathcal{T}=\emptyset$ and there exists only one irreducible closed set of persistent states.
 	 
 	 Assume, for the sake of contradiction, that  $\mathcal{G}_{\mathcal{D}}$ is connected and there exists a state $j$ outside the set $\mathcal{C}$. This, by definition, means that there exists no path that starts in state $i$ and ends in state $j$. 
 	 
 	 Here we will show that when $\mathcal{G}_{\mathcal{D}}$ is connected, it is always possible to construct such a path. Let $v_i$ be the item corresponding to state $i$ and $v_j$ the item corresponding to state $j$. Let $\mathfrak{D}_{v_i}$ the proximal set of items of $v_i$. We must have one of the following cases: 
 	 \begin{itemize}
 	 	\item[$v_j \in \mathfrak{D}_{v_i}$:] In this case, the states are directly connected, and $\Pr\{\text{next is $j$}|\text{we are in $i$}\}$ equals:
 	 	\begin{equation}
 	 	 [\mathbf{D}]_{ij}= \sum_{\mathcal{D}_k \in \mathfrak{D}_{v_i}, v_j \in \mathcal{D}_k}\frac{1}{N_{v_i}\lvert \mathcal{D}_k\rvert}
 	 	 \label{probability_d_ij}
 	 	\end{equation}
 	 	which can be seen by Eq.~(\ref{def:D}) together with the definitions of Section \ref{ModelDefinitions}. 
 	 	\item[$v_j \notin \mathfrak{D}_{v_i}$:] In this case, the states are not directly connected. Let $\mathcal{D}_{v_j}$ be a $\mathcal{D}$-block that contains $v_j$, and $\mathcal{D}_{v_i}$ a $\mathcal{D}$-block that contains $v_i$. Notice that $v_j \notin \mathfrak{D}_{v_i}$ implies that $\mathcal{D}_{v_i}\cap\mathcal{D}_{v_j}=\emptyset$. However, since $\mathcal{G}_{\mathcal{D}}$ is assumed connected, there exists a sequence of vertices corresponding to $\mathcal{D}$-blocks, that forms a path in the block coupling graph between nodes $\mathcal{D}_{v_i}$ and $\mathcal{D}_{v_j}$. Let this sequence be the one below:
 	 	\begin{equation}
 	 	\mathcal{D}_{v_i},\mathcal{D}_1,\mathcal{D}_2,\dots,\mathcal{D}_n,\mathcal{D}_{v_j}
 	 	\end{equation}  Then, choosing arbitrarily one state that corresponds to an item belonging to each of the $\mathcal{D}$-blocks of the above sequence, we get the sequence of states:
 		\begin{equation}
 		i,t_1,t_2,\dots,t_n,j
 		\label{StateSequence}
 		\end{equation}
 		which corresponds to the sequence of items 
 		\begin{equation}
 		v_i,v_{t_1},v_{t_2},\dots,v_{t_n},v_j
 		\end{equation}
 		Notice that the definition of the $\mathcal{D}$-blocks together with the definitions of the proximal sets and the block coupling graph, imply that this sequence has the property every item, after $v_i$, to belong to the proximal set of the item preceding it, i.e.  
 		\begin{equation}
 		v_{t_1} \in \mathfrak{D}_{v_i},v_{t_2} \in \mathfrak{D}_{v_{t_1}},\dots,v_{j} \in \mathfrak{D}_{v_{t_n}}
 		\end{equation}
 		Thus, the consecutive states in sequence~(\ref{StateSequence}) communicate, or
 		\begin{equation}
 		i\to t_1 \to t_2 \to \dots \to t_n \to j
 		\label{trajectory}
 		\end{equation}
 		and there exists a positive probability path between states $i$ and $j$.
     \end{itemize}
     
 	 In concussion, when $\mathcal{G}_{\mathcal{D}}$ is connected there will always be a path starting from state $i$ and ending in state $j$. But because state $i$ is persistent, and belongs to the irreducible closed set of states $\mathcal{C}$, state $j$ belongs to the same irreducible closed set of states too. This contradicts our assumption. Thus, when $\mathcal{G}_{\mathcal{D}}$ is connected every state belongs to a single irreducible closed set of states, $\mathcal{C}$. 
 \end{proof}
  
  Now it remains to prove the aperiodicity property.
  
  \begin{lemma}
  	The Markov chain induced by matrix $\mathbf{D}$ is aperiodic.  
  \end{lemma}
   \begin{proof}
	It is known that the period of a state $i$ is defined as the greatest common divisor of the epochs at which a return to the state is possible~\cite{grimmett2001probability}. Thus, it suffices to show that we can return to any given state in consecutive time epochs. But this can be seen readily because the diagonal elements of matrix $\mathbf{D}$ are by definition, all greater than zero; thus, for any state and for every possible trajectory  of the Markov chain of length $k$ there is another one of length $k+1$ with the same starting and ending state, that follows the self loop as its final step. In other words, leaving any given state of the corresponding Markov chain, one can always return in consecutive time epochs, which makes the chain aperiodic. And the proof is complete.
   \end{proof}
  
 We have shown so far that the connectivity of $\mathcal{G}_{\mathcal{D}}$ results is enough to ensure the irreducibility and aperiodicity of the Markov chain with transition probability matrix $\mathbf{D}$.
  
  It remains now to prove that the same thing holds for the complete stochastic matrix $\mathbf{S}$. This can be done using the following useful lemma, the proof of which can be found in the Appendix \ref{Ap:LemmaProof}.
\begin{lemma}
	If $\mathbf{A}$ is the transition matrix of an irreducible and aperiodic Markov chain with finite state space, and $\mathbf{B}$ the transition matrix of any Markov chain defined onto the same state space, then matrix $\mathbf{C} = \kappa\mathbf{A} + \lambda\mathbf{B}$, where $\kappa,\lambda > 0$ such that $\kappa+\lambda = 1$ denotes the transition matrix of an irreducible and aperiodic Markov chain also.
	\label{Lemma4}
\end{lemma}
Applying Lemma~\ref{Lemma4} twice, first to matrix:
\begin{equation}
\mathbf{T} = \beta \mathbf{H} + (1-\beta)\mathbf{D}
\end{equation} 
and then to matrix:
\begin{equation}
\mathbf{S} = (1-\alpha)\mathbf{e}\boldsymbol{\omega}^\intercal + \alpha\mathbf{T}
\end{equation} 
gives us the irreducibility and the aperiodicity of matrix $\mathbf{S}$. Taking into account the fact that the state space is finite, the resulting Markov chain becomes ergodic~\cite{grimmett2001probability} and there exists a unique recommendation vector corresponding to its steady state probability distribution which is given by 
\begin{equation}
\boldsymbol{\pi} = [\pi_1 \pi_2 \cdots \pi_m] = [\frac{1}{\mu_1} \frac{1}{\mu_2} \cdots \frac{1}{\mu_m}]
\end{equation}
where $\mu_i$ is the mean recurrence time of state $i$. However, for ergodic states, by definition it holds that 
\begin{equation}
1\leq \mu_i < \infty
\end{equation}
Thus $\pi_i>0$, for all $i$, and the support of the distribution that defines the recommendation vector includes every item of the  underlying space.

\end{proof}

The above theorem suggests that even for a user who have rated only one item, when the chosen decomposition enjoys the above property, our recommender finds a way to assign preference probabilities for the complete item space. Note that the criterion for this to be true is not that restrictive. For example for the  \texttt{MovieLens} datasets, using as a criterion of decomposition the categorization of movies into genres, the block coupling graph is connected. This, proves to be a very useful property, in dealing with the cold-start problem as we will see in the experimental evaluation presented in Section~\ref{SubSec:NewUsersTests}.

\subsection{NCDREC Algorithm: Storage and Computational Issues}
\label{SubSec:NCDREC_Algorithm}

\begin{algorithm*}
	\begin{algorithmic}
		\begin{multicols}{2}
			\State \textbf{Input:} Matrices $\mathbf{R}\in\mathbb{R}^{n \times m}$, $\mathbf{H} \in \mathbb{R}^{m \times m}, \mathbf{X} \in \mathbb{R}^{m\times K}, \mathbf{Y} \in  \mathbb{R}^{K\times m},  \mathbf{Z} \in  \mathbb{R}^{n\times K} $. Parameters $\alpha, \beta ,f, \epsilon$ \\
			\textbf{Output:} The matrix with recommendation vectors for every user, $\mathbf{\Pi}\in \mathbb{R}^{n \times m}$
			\State \textbf{Step 1:} Find the newly added users and collect their preference vectors into matrix $\mathbf{\Omega}$. 
			\State \textbf{Step 2:} Compute $\mathbf{\Pi}_{\text{sparse}}$ using the {\sc ColdStart} procedure.
			\State \textbf{Step 3:} Initialize vector $\mathbf{p_1}$ to be a random unit length vector.
			\State \textbf{Step 4:} Compute the modified Lanczos procedure up to step $M$, using {\sc NCD\_PartialLBD}  with starting vector $\mathbf{p_1}$.
			\State \textbf{Step 5:} Compute the SVD of the bidiagonal matrix $\mathbf{B}$ to extract $f<M$ approximate singular triplets: 
			\begin{displaymath}
			\{\mathbf{\tilde{u}_j}, \sigma_j, \mathbf{\tilde{v}_j} \} \gets \{ \mathbf{Q}\mathbf{u_j^{(B)}}, \sigma_j^{(B)}, \mathbf{P}\mathbf{v_j^{(B)}}\}
			\end{displaymath}
			\State \textbf{Step 6:} Orthogonalize against the approximate singular vectors to get a new starting vector $\mathbf{p_1}$.
			\State \textbf{Step 7:} Continue the Lanczos procedure for $M$ more steps using the new starting vector.
			\State \textbf{Step 8:} Check for convergence tolerance. If met compute matrix $\mathbf{\Pi}_{\text{full}} = \mathbf{\tilde{U}\Sigma\tilde{V}^\intercal}$ else go to \textbf{Step 4}.
			\State \textbf{Step 9:} Update $\mathbf{\Pi}_{\text{full}}$, replacing the rows that correspond to new users with  $\mathbf{\Pi}_{\text{sparse}}$.
			\State \Return $\mathbf{\Pi}_{\text{full}}$
			\Statex
			\Procedure{NCD\_PartialLBD}{$\mathbf{R},\mathbf{X},\mathbf{Z},\mathbf{p_1},\epsilon$}
			\State $\boldsymbol{\phi} \gets \mathbf{X}^\intercal\mathbf{p_1}$; $\mathbf{q_1} \gets \mathbf{R}\mathbf{p_1} + \epsilon\mathbf{Z}\boldsymbol{\phi}$;
			\State $b_{1,1} \gets \lVert \mathbf{q_1} \rVert_{2}$ ; $\mathbf{u_1} \gets \mathbf{q_1}/b_{1,1}$;
			\For{$j = 1 \text{ to } M$}
			\State $\boldsymbol{\phi} \gets \mathbf{Z}^{\intercal}\mathbf{q_j}$;
			\State $\mathbf{r} \gets \mathbf{R}^{\intercal}\mathbf{q_j} + \epsilon\mathbf{X}\boldsymbol{\phi} - b_{j,j}\mathbf{p_j}$;
			\State $\mathbf{r} \gets \mathbf{r} - \left[\mathbf{p_{1}} \dots \mathbf{p_{j}}\right]\left([\mathbf{p_{1}} \dots \mathbf{p_{j}}]^{\intercal}\mathbf{r} \right)$;
			\If{$j < M$} 
			\State $b_{j,j+1} \gets \lVert \mathbf{r} \rVert$; $\mathbf{p_{j+1}} \gets \mathbf{r}/b_{j,j+1}$;
			\State $\boldsymbol{\phi} \gets \mathbf{X}^\intercal\mathbf{p_{j+1}}$;
			\State $\mathbf{q_{j+1}} \gets \mathbf{R}\mathbf{p_{j+1}} + \epsilon\mathbf{Z}\boldsymbol{\phi} - b_{j,j+1} \mathbf{q_j}$;
			\State $\mathbf{q_{j+1}} \gets \mathbf{q_{j+1}} - \left[\mathbf{q_{1}} \dots \mathbf{q_{j}}\right]([\mathbf{q_{1}} \dots \mathbf{q_{j}}]^{\intercal}\mathbf{q_{j+1}})$;
			\State $b_{j+1,j+1} \gets \lVert \mathbf{q_{j+1}} \rVert$;
			\State $\mathbf{q_{j+1}} \gets \mathbf{q_{j+1}}/b_{j+1,j+1}$;
			\EndIf
			\EndFor
			\EndProcedure
			\Statex
			\Procedure{ColdStart}{$\mathbf{H},\mathbf{X},\mathbf{Y},\mathbf{\Omega},\alpha,\beta$}
			\State $\mathbf{\Pi} \gets \mathbf{\Omega}$;  $k\gets0$; $r \gets 1;$
			\While{$ r > \mathrm{tol}$ and $k\leq \mathrm{maxit}$}
			\State $k \gets k+1$; 
			\State $\mathbf{\hat{\Pi}} \gets \alpha\beta \mathbf{\Pi}\mathbf{H}$; $\mathbf{\Phi} \gets \mathbf{\Pi} \mathbf{X} $; 
			\State $\mathbf{\hat{\Pi}} \gets \mathbf{\hat{\Pi}} + \alpha(1-\beta)\mathbf{\Phi}\mathbf{Y} +(1-\alpha)\mathbf{\Omega}$; 
			\State $ r \gets \lVert \mathbf{\hat{\Pi}} - \mathbf{\Pi}\rVert$; $ \mathbf{\Pi} \gets \mathbf{\hat{\Pi}} $;
			\EndWhile \\
			\Return $\mathbf{\Pi}_{\text{sparse}} \gets \mathbf{\Pi} $
			\EndProcedure
		\end{multicols}
	\end{algorithmic}
	\caption{NCDREC Algorithm}\label{NCDREC}
\end{algorithm*}

It is clear that for the majority of reasonable decompositions the number of blocks is much smaller than the cardinality of the item space, i.e. $K \ll m$; this makes matrices $\mathbf{D}$ and $\mathbf{W}$, extremely low-rank. Thus, if we take into account the inherent sparsity of the ratings matrix $\mathbf{R}$, and of the component matrices $\mathbf{X,Y,Z}$,  we see that the storage needs of NCDREC are in fact modest. 

Furthermore, the fact that matrices  $\mathbf{G}$ and $\mathbf{S}$ can be expressed as a sum of sparse and low-rank components, can also be exploited computationally as we see in the NCDREC algorithm presented above. Our algorithm makes sure that the computation of the recommendation vectors can be carried out without needing to explicitly compute matrices $\mathbf{G}$ and $\mathbf{S}$.

The computation of the singular triplets is based on a fast partial SVD method proposed by Baglama and Reighel in \cite{baglama2005augmented}. However, because their method presupposes the existence of the final matrix, we modified the partial Lanczos bidiagonalization iterative procedure to take advantage of the factorization of the NCD preferences matrix $\mathbf{W}$ into matrices $\mathbf{X},\mathbf{Z}$. The detailed computation is presented in the \textsc{NCD\_PartialLBD} procedure in Algorithm~\ref{NCDREC}.  For the computation of the newly added users' recommendations, we collect their preference vectors in an extremely sparse matrix $\mathbf{\Omega}$, and we compute their stationary distributions using a batch power method approach exploiting matrices $\mathbf{X},\mathbf{Y}$.
Notice that the exploitation of the factorization of the NCD matrices in both procedures results in a significant drop of the number of floating point operations per iteration, since every dense Matrix$\times$Vector (MV) multiplication, is now replaced by a sum of lower dimensional and sparse MV's, making the overall method significantly faster.

\section{Experimental Evaluation}
\label{Sec_Experiments}
In order to evaluate the performance of NCDREC in recommending top-N lists of items, we run a number of experiments using two real datasets: the  \texttt{Yahoo!R2Music}, which  represents a real snapshot of the Yahoo! Music community's preferences for various songs, and the standard \texttt{MovieLens} (1M and 100K) datasets. These datasets also come with information that relates the items to genres; this was chosen as the criterion of decomposition behind the definition of matrices $\mathbf{D}$ and $\mathbf{W}$. 
 For further details about these datasets see  \url{http://webscope.sandbox.yahoo.com} and \url{http://grouplens.org/}. A synopsis of their basic characteristics is presented in Table~\ref{table:Dataset}.

Exploiting meta-information is a very useful weapon in alleviating sparsity related problems~\cite{DBLP:reference/rsh/DesrosiersK11}.  Thus, in order to provide fair comparisons we test our method against recommendation methods that:
\begin{itemize}
	\item[(a)] can also take advantage of the categorization of items to genres and,
	\item[(b)] are known to show lower sensitivity to the problems of limited coverage and sparsity \cite{DBLP:reference/rsh/DesrosiersK11}.
	\end{itemize}
	  In particular, we run NCDREC\footnote{The perturbation parameter $\epsilon$ was set to 0.01, the number of latent factors was selected from the range 2 to 800, and the \textsc{ColdStart} subcomponent parameters were chosen to be $\alpha = 0.01$ and $\beta = 0.75$.} against five state-of-the-art graph-based approaches; the node similarity algorithms $\mathbf{L}^\dagger$, and \textbf{Katz}; the random walk approaches \textbf{First Passage Time} (FP) and \textbf{Commute Time} (CT) and the \textbf{Matrix Forest Algorithm} (MFA).

\subsection{Competing Recommendation Methods}
\label{subsec:MovieLensDataset} 

The data model used for all the competing methods is a graph representation of the recommender system database. Concretely, consider a weighted graph $G$ with nodes corresponding to database elements and database links corresponding to edges.
For example, in the \texttt{MovieLens} datasets each element of the \texttt{people} set, the \texttt{movie} set,
and the \texttt{movie\_category} set, corresponds to a node of the
graph, and each \texttt{has\_watched} and \texttt{belongs\_to} link is 
expressed as an edge \cite{Fouss:2012:EIK:2207289.2207578,Fouss:2007:RCS:1263132.1263335}.

Generally speaking, graph-based recommendation methods work by computing similarity measures between every element in the recommender database and then using these measures to compute ranked lists of the items with respect to each user. 

\begin{table*}[!hpbt]
	\centering
	\setlength{\tabcolsep}{1.51em}
	\caption{Datasets}
	{\normalsize \begin{tabular}{rrrrr} 
			\toprule
			\toprule
			Dataset &  \#Users & \#Items & \#Ratings & Density \\
			\midrule
			\texttt{MovieLens100K} & 943 & 1682 & 100,000 & 6.30\% \\
			\texttt{MovieLens1M} & 6,040 & 3,883 & 1,000,209 & 4.26\% \\
			\texttt{Yahoo!R2Music} & 1,823,179 & 136,736 & 717,872,016 & 0.29\% \\
			\bottomrule
			\bottomrule
		\end{tabular}}
		\label{table:Dataset}
	\end{table*}

\paragraph{The pseudoinverse of the graph's Laplacian ($\mathbf{L\dagger}$).} This matrix contains the inner products of the node vectors
in a Euclidean space where the nodes are exactly separated by the commute time distance \cite{Fouss:2007:RCS:1263132.1263335}. 
For the computation of the $\mathbf{L^\dagger}$ matrix we used the formula: 
\begin{equation}
\mathbf{L^\dagger} = (\mathbf{L}-\frac{1}{n+m+K}\mathbf{e}\mathbf{e}^\intercal)^{-1}+\frac{1}{n+m+K}\mathbf{e}\mathbf{e}^\intercal
\end{equation} where $\mathbf{L}$ is the Laplacian of the graph model of the recommender system, $n$, the number of users, $m$, the number of items, and $K$, the number of blocks (see~\cite{Fouss:2012:EIK:2207289.2207578} for details).
\paragraph{The MFA similarity matrix $\mathbf{M}$.} MFA matrix contains elements that also provide similarity measures between nodes of the graph by integrating indirect paths, based on the matrix-forest theorem~\cite{chebotarev2006matrix}. Matrix $\mathbf{M}$ was computed by 
\begin{equation}
\mathbf{M} = \left(\mathbf{I}+\mathbf{L}\right)^{-1}
\end{equation}
where $\mathbf{I}$, the identity matrix.
\paragraph{The Katz similarity matrix $\mathbf{K}$.} Katz similarity matrix is computed by
\begin{equation}
\mathbf{K} = \alpha\mathbf{A}+\alpha^2\mathbf{A}^2+\dots = \left(\mathbf{I-\alpha\mathbf{A}}\right)^{-1}-\mathbf{I}
\end{equation}
where $\mathbf{A}$ is the adjacency matrix of the graph and $\alpha$ measures the attenuation in a link (see~\cite{KATZ}). 

\paragraph{Average First Passage Times.} The Average First Passage Time scores are computed by iteratively solving the recurrence
\begin{equation}
\left\{\begin{matrix}
\operatorname{FP}(k|k) &=& 0& &\\
\operatorname{FP}(k|i) &=& 1&+&\sum_{j=1}^{n+m+K}p_{ij}\operatorname{FP}(k|j), \quad \text{for }i\ne k
\end{matrix} \right.
\end{equation} where $p_{ij}$ is the conditional probability a random walker in the graph $G$, visits node  $j$ next, given that he is currently in node $i$. 

\paragraph{Average Commute Times.} Finally, Average Commute Times scores can be obtained in terms of the Average First-Passage Times by:
\begin{equation}
\operatorname{CT}(i,j) = \operatorname{FP}(i|j)+\operatorname{FP}(j|i)
\end{equation}

For further details about the competing algorithms  see \cite{Fouss:2007:RCS:1263132.1263335,Fouss:2012:EIK:2207289.2207578,chebotarev2006matrix,KATZ} and the references therein.  

\begin{figure*}[!htb]
\centering
\epsfig{file=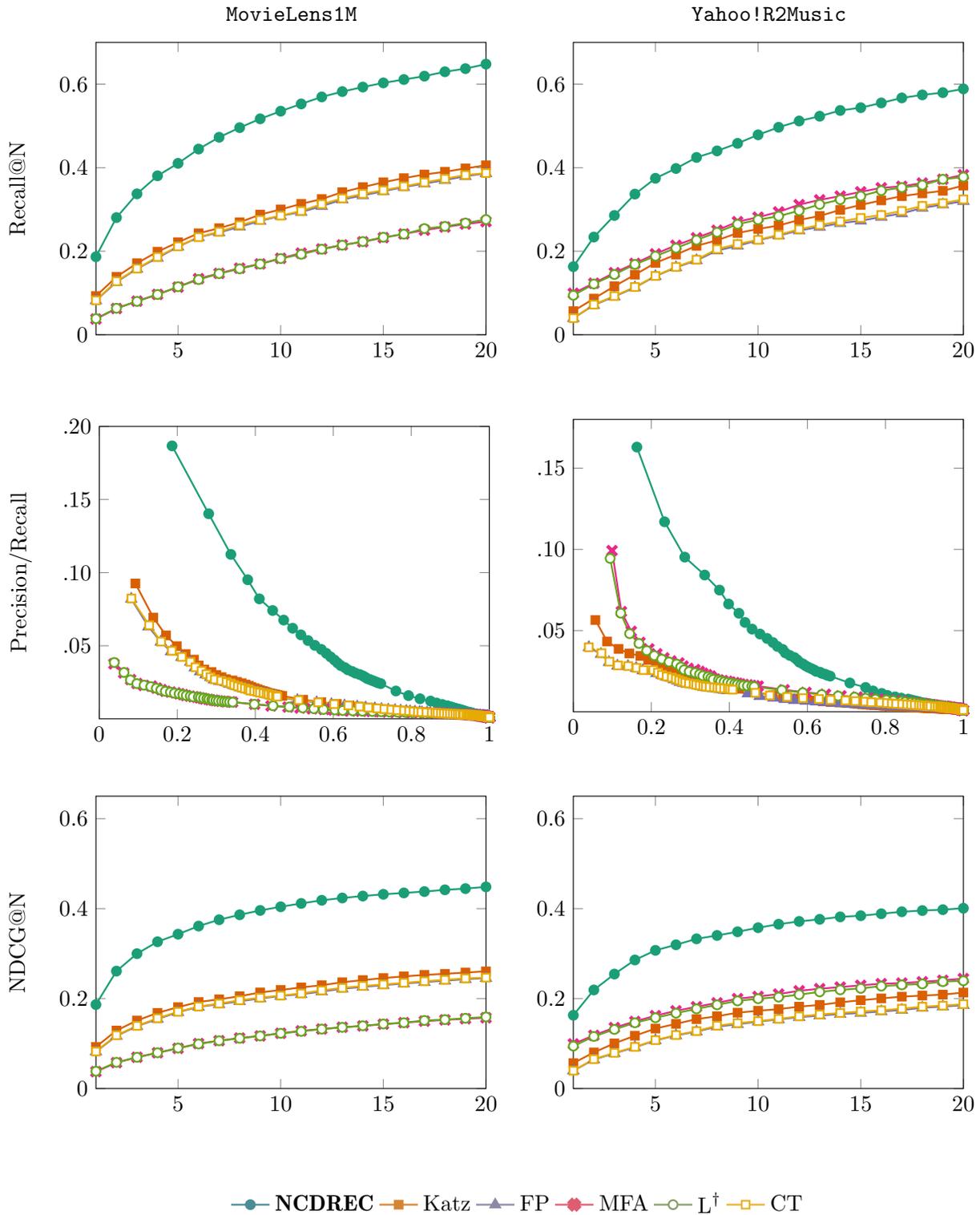,width=\textwidth}
\caption{Recommendation quality on \texttt{MovieLens1M} and \texttt{Yahoo!R2Music} datasets using Recall@N, Precision and NDCG@N metrics}
\label{fig:QualityFull}
\end{figure*}

\subsection{Quality of Recommendation} 
\label{SubSec:StandardRecommendation}
To evaluate the quality of our method in suggesting top-N items,
we have adopted the methodology used in~\cite{Cremonesi:2010:PRA:1864708.1864721}. In particular, we randomly sampled 1.4\% of the ratings of the dataset in order to create a probe set $\mathcal{P}$, and we use each item $v_{j}$, rated with 5-star by user $u_i$ in $\mathcal{P}$ to form the test set $\mathcal{T}$. Finally, for each item in $\mathcal{T}$, we randomly select another 1000 unrated items of the same user and we rank the 1001 item lists using the different methods mentioned and we evaluate the quality of recommendations. 

For this evaluation, except for the
standard \textbf{Recall} and \textbf{Precision}  metrics~\cite{Baeza-Yates:2008:MIR:1796408,Cremonesi:2010:PRA:1864708.1864721}, we
also use a number of other well known ranking measures, which discount the utility of recommended items depending on their position in the recommendation list~\cite{Balakrishnan:2012:CR:2124295.2124314,DBLP:reference/rsh/ShaniG11}; namely the \textbf{R-Score}, the \textbf{Normalized Discounted Cumulative Gain} and the \textbf{Mean Reciprocal Rank} metrics. R-Score assumes that the value of recommendations decline \textit{exponentially} fast to yield for each user the following score:
\begin{equation}
\mathrm{R}(\alpha) = \sum_{q} \frac{\max(y_{\pi_q}-d,0)}{2^{\frac{q-1}{\alpha-1}}}
\end{equation}
where $\alpha$ is a half-life parameter which controls the exponential decline, $\pi_q$ is the index of the $q^{th}$ item in the recommendation ranking list $\boldsymbol{\pi}$, and $\mathbf{y}$ is a vector of the relevance values for a sequence of items. In Cumulative Discounted Gain the ranking positions are discounted \textit{logarithmically} and is defined as:
\begin{equation}
\mathrm{DCG@}k(\mathbf{y},\boldsymbol{\pi}) = \sum_{q=1}^{k} \frac{2^{y_{\pi_q}} - 1}{\log_2(2+q)} 
\end{equation}
The Normalized Discounted Cumulative Gain can then be defined as:
\begin{equation}
\mathrm{NDCG@}k = \frac{\mathrm{DCG@}k(\mathbf{y},\boldsymbol{\pi})} {\mathrm{DCG@}k(\mathbf{y},\boldsymbol{\pi^\star})} 
\end{equation}
where, $\boldsymbol{\pi^\star}$ is the best possible ordering of the items with respect to the relevant scores (see~\cite{Balakrishnan:2012:CR:2124295.2124314} for details). Finally, Mean Reciprocal Rank (MRR) is the average of each user's reciprocal rank score, defined as follows: 
\begin{equation}
\mathrm{RR} =  \frac{1}{\min_q \{q:y_{\pi_q}>0\} }
\end{equation}
MRR decays more slowly than R-Score but faster than NDCG.

Figure \ref{fig:QualityFull} reports the performance of the algorithms on the Recall, Precision and NDCG metrics. In particular, we report the average Recall as a function of $N$ (focusing on the range $N=[1,\dots,20]$), the Precision at a given Recall, and the NDCG@$N$, for the \texttt{MovieLens1M} (1st column) and \texttt{Yahoo!R2Music} (2nd column) datasets. As we can see NCDREC outperforms all other methods, reaching for example at $N=10$, a Recall around 0.53 on \texttt{MovieLens} and 0.45 on the sparser \texttt{Yahoo!R2Music} dataset. Similar behavior is observed for the Precision and the NDCG metrics as well. Table \ref{Table:FullQuality} presents the results for the R-Score (with halflife parameters 5 and 10) and the MRR metrics. Again we see that NCDREC achieves the best results with MFA and $L\dagger$ doing significantly better than the other graph-based approaches in the sparser dataset. 

\begin{table*}[htbp!] 
	\centering
	\setlength{\tabcolsep}{1.51em}
	\caption{Recommendation quality on \texttt{MovieLens1M} and \texttt{Yahoo!R2Music} datasets using R-Score and MRR metrics}
	{\normalsize \begin{tabular}{@{}rrrrcrrrc@{}}
			\toprule
			\toprule
			&
			\multicolumn{3}{c}{\textbf{\texttt{MovieLens1M}}} &
			\phantom{abc}&
			\multicolumn{3}{c}{\textbf{\texttt{Yahoo!R2Music}}} &
			\\
			\cmidrule{2-4}
			\cmidrule{6-8}
			& R(5) & R(10) & MRR & & R(5) & R(10) & MRR  \\ 
			\midrule
			NCDREC & \textbf{0.3997} & \textbf{0.5098} & \textbf{0.3008} & & \textbf{0.3539} & \textbf{0.4587} & \textbf{0.2647} \\                      
			MFA & 0.1217 & 0.1911 & 0.0887 & & 0.2017 & 0.2875 & 0.1591  \\
			$L\dagger$ & 0.1216 & 0.1914 & 0.0892 & & 0.1965 & 0.2814 & 0.1546  \\
			FP & 0.2054 & 0.2874 & 0.1524 & & 0.1446 & 0.2241 & 0.0998  \\
			Katz & 0.2187 & 0.3020 & 0.1642 & & 0.1704 & 0.2529 & 0.1203  \\ 
			CT & 0.2070 & 0.2896 & 0.1535 & & 0.1465 & 0.2293 & 0.1019  \\
			\bottomrule
			\bottomrule
		\end{tabular}}
		\label{Table:FullQuality}
	\end{table*}

Finally, for completeness, we also run NCDREC on the standard \texttt{MovieLens100K} dataset using the publicly available 5 predefined splittings into training and test sets. Here, we use the Degree of Agreement metric (a variant of Somer's D statistic\footnote{We give a detailed definition of the DOA metric in Section~\ref{SubSec:NewUsersTests} where we also present other ranking stability metrics.}, that have been used by many authors for the performance evaluation of ranking-based recommendations on \texttt{MovieLens100K}) in order to allow direct comparisons with the different results to be found in the literature \cite{Fouss:2007:RCS:1263132.1263335,Freno:2009:SPE:1557019.1557059,Gori:2007:IRB:1625275.1625720,iExpand,Zhang:2008:TPB:1390334.1390465}. 

NCDREC obtained a macro-averaged DOA score of \textbf{92.25} and a micro-averaged DOA of \textbf{90.74} which is -- to the best of our knowledge -- the highest scores achieved thus far on this benchmark dataset.

\subsection{Long-Tail Recommendation}
\label{SubSec:LongTailTests}
It is well known that the distribution of rated items in recommender systems is long-tailed, i.e. the majority of the ratings is concentrated in a few very popular items. Of course, recommending popular items is generally considered an easy task and adds very little utility in recommender systems. On the other hand, the task of recommending long-tail items adds \textit{novelty} and \textit{serendipity} to the users~\cite{Cremonesi:2010:PRA:1864708.1864721}, and it is also known to increase the profits of e-commence companies significantly~\cite{anderson2008long,yin2012challenging}. The inherent sparsity of the data however -- which is magnified even more for long tail items -- presents a major challenge for most state-of-the-art collaborative filtering methods.

\begin{figure*}[!htbp]
\centering
\epsfig{file=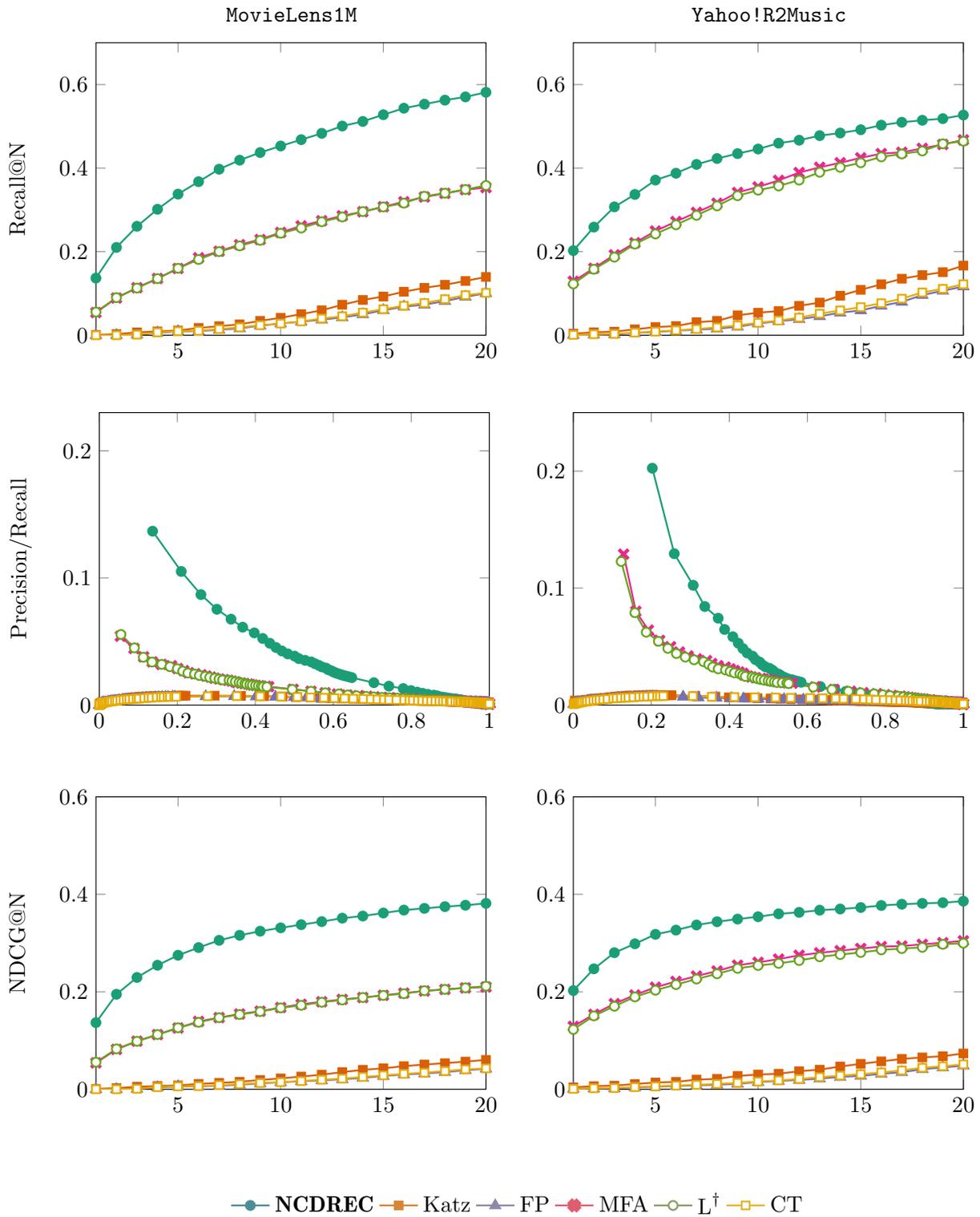,width=\textwidth}
\caption{Long tail recommendation quality on \texttt{MovieLens1M} and \texttt{Yahoo!R2Music} datasets using Recall@N, Precision and NDCG@N metrics}
\label{fig:QualityLong}
\end{figure*}

\begin{table*}[!hpbt]
	\centering
	\setlength{\tabcolsep}{1.51em}
	\caption{Long tail recommendation quality on \texttt{MovieLens1M} and \texttt{Yahoo!R2Music} datasets using R-Score and MRR metrics}
	{\normalsize \begin{tabular}{@{}rrrrcrrrc@{}}
			\toprule
			\toprule
			&
			\multicolumn{3}{c}{\textbf{\texttt{MovieLens1M}}} &
			\phantom{abc}&
			\multicolumn{3}{c}{\textbf{\texttt{Yahoo!R2Music}}} &
			\\
			\cmidrule{2-4}
			\cmidrule{6-8}
			& R(5) & R(10) & MRR & & R(5) & R(10) & MRR \\ 
			\midrule
			NCDREC & \textbf{0.3279} & \textbf{0.4376} & \textbf{0.2395} & & \textbf{0.3520} & \textbf{0.4322} & \textbf{0.2834}  \\                      
			MFA & 0.1660 & 0.2517 & 0.1188 & & 0.2556 & 0.3530 & 0.1995  \\
			$L\dagger$ & 0.1654 & 0.2507 & 0.1193 & & 0.2492 & 0.3461 & 0.1939  \\
			FP & 0.0183 & 0.0654 & 0.0221 & & 0.0195 & 0.0684 & 0.0224  \\
			Katz & 0.0275 & 0.0822 & 0.0267 & & 0.0349 & 0.0939 & 0.0309  \\
			CT & 0.0192 & 0.0675 & 0.0227 & & 0.0215 & 0.0747 & 0.0249  \\
			\bottomrule
			\bottomrule
		\end{tabular}}
		\label{Table:LongQuality}
	\end{table*}

In order to evaluate NCDREC in recommending long-tail items we adopt the methodology described in~\cite{Cremonesi:2010:PRA:1864708.1864721}. In particular, we order the items according to their popularity (the popularity was measured in terms of number of ratings) and we further partition the test set $\mathcal{T}$ into two subsets, $\mathcal{T}_{\mathrm{head}}$ and $\mathcal{T}_{\mathrm{tail}}$, that involve items originated from the short head, and the long tail of the distribution, respectively. We discard the popular items and we evaluate NCDREC and the other algorithms on the  $\mathcal{T}_{\mathrm{tail}}$ test set, using the procedure explained in the previous section. Figure~\ref{fig:QualityLong} and Table~\ref{Table:LongQuality} report the results.

We see that NCDREC achieves again the best results, managing to retain its performance in all metrics and for both datasets. Notice here the significant drop in quality of the random walk based methods, which were found to behave very well in the standard recommendation scenario. This finding indicates their bias in recommending popular items.  MFA and $L\dagger$ on the other hand, do particularly well, exhibiting great ability in uncovering non-trivial relations between the items, especially in the sparser \texttt{Yahoo!R2Music} dataset.

\subsection{Recommendations for Newly Emerging Users} 
\label{SubSec:NewUsersTests}

One very common manifestation of sparsity faced by real recommender systems is the \textit{New-Users Problem}. This problem refers to the difficulty of achieving reliable recommendations for newly emerging users in an existing recommender system, due to the \textit{de facto} initial lack of personalized feedback. This problem can also be seen as an extreme and localized expression of sparsity, that prohibits CF methods to uncover meaningful relations between the set of new users and the rest of the RS database, and thus, undermines the reliability of the produced recommendations.  

To evaluate the performance of our method in coping with this problem we run the following experiment. We randomly select 100 users from the \texttt{MovieLens1M} dataset having rated 100 movies or more and we randomly select to include 4\%, 6\%, 8\%, 10\% of their ratings in new artificially ``sparsified'' versions of the dataset. The idea is that the modified data represent ``earlier snapshots'' of the system, when these users were new, and as such, had rated fewer items. We run NCDREC\footnote{Note that the ranking list for the set of newly added users was produced by the \textsc{ColdStart} subcomponent.} against the other methods, and we compare the recommendation vectors with the ranking lists induced by the complete set of ratings, which we use as the reference ranking for each user.

For this comparison except for the standard \textbf{Spearman's $\boldsymbol{\rho}$} and \textbf{Kendall's $\boldsymbol{\tau}$} metrics \cite{Baeza-Yates:2008:MIR:1796408,DBLP:reference/rsh/ShaniG11}, we also use two other  well known ranking measures, namely the
\textbf{Degree of Agreement} (DOA) \cite{Fouss:2007:RCS:1263132.1263335,Freno:2009:SPE:1557019.1557059,Gori:2007:IRB:1625275.1625720} and the \textbf{Normalized Distance-based Performance Measure} (NDPM) \cite{DBLP:reference/rsh/ShaniG11}, outlined below. Table~\ref{table:Metrics} contains all the necessary definitions. 

\begin{table*}[!htpb]
	\centering
	\caption{A summary of the notation used for the definition of the ranking stability metrics}
	\ra{1.19}
	{\normalsize \begin{tabular}{rl} 
			\toprule
			\toprule
			Notation & Meaning \\
			\midrule
			$\boldsymbol{r^i}$ & User's $u_i$ reference ranking \\
			$\boldsymbol{\pi^i}$ & Recommender System generated ranking\\
			$r^i_{v_j}$& Ranking score of the item $v_j$ in user's $u_i$ ranking list (reference ranking)\\
			$\pi^i_{v_j}$& Ranking score of the item $v_j$ in user's $u_i$ ranking list (Recommender System generated ranking)\\ 
			$C$& Number of pairs that are concordant\\
			$D$& Number of discordant pairs\\
			$N$& Total number of pairs\\
			$T_r$& Number of tied pairs in the reference ranking\\
			$T_\pi$& Number of tied pairs in the system ranking\\
			$X$ & Number of pairs where the reference ranking does not tie, but the RS's \\
			& ranking ties ($N-T_r-C-D$) \\
			\bottomrule
			\bottomrule
		\end{tabular}}
		\label{table:Metrics}
	\end{table*}

	\begin{description}
		\item [Kendall's $\boldsymbol{\tau}$] is an intuitive nonparametric rank correlation index that has been widely used in the literature. The $\tau$ of ranking lists $\boldsymbol{r^i}$, $\boldsymbol{\pi^i}$ is defined to be:
		\begin{equation}
		\tau \triangleq \frac{C- D}{\sqrt{N - T_r} \sqrt{N - T_\pi}}
		\end{equation}
		and takes the value of 1 for perfect match and -1 for reversed ordering.
		
		\item [Spearman's $\boldsymbol{\rho}$] is another widely used non-parametric measure of rank correlation. The $\rho$ of ranking lists $\boldsymbol{r^i}$, $\boldsymbol{\pi^i}$ is defined to be:
		\begin{equation}
		\rho \triangleq \frac{1}{m} \frac{\sum_{v_j}(r^i_{v_j}-\bar{r}^i) (\pi^i_{v_j}-\bar{\pi}^i)}{\sigma(\boldsymbol{r^i}) \sigma(\boldsymbol{\pi^i})} 
		\end{equation}
		where the $\bar{\cdot}$ and $\sigma(\cdot)$ denote the mean and standard deviation. 
		The $\rho$ takes values from -1 to 1. A $\rho$ of 1 indicates perfect rank association, a $\rho$ of zero indicates no association between the ranking lists and a $\rho$ of -1 indicate a perfect negative association of the rankings. 
		\item[Degree of Agreement] (DOA) is a performance index commonly used in the recommendation literature to evaluate the quality of ranking-based CF methods \cite{Fouss:2007:RCS:1263132.1263335,Freno:2009:SPE:1557019.1557059,Gori:2007:IRB:1625275.1625720,Zhang:2008:TPB:1390334.1390465}. DOA is a variant of the Somers' D statistic \cite{siegel1988nonparametric}, defined as follows:
		\begin{equation}
		\text{DOA}_i \triangleq \frac{\sum_{v_j \in \mathcal{T}_i\wedge v_k \in \mathcal{\overline{W}}_i}[\pi^i_{v_j}>\pi^i_{v_k}]}{\mid\mathcal{T}_i\mid*\mid \overline{(\mathcal{L}_i \cup \mathcal{T}_i)} \mid}
		\end{equation} 
		where $\left[S\right]$ equals 1, if statement $S$ is true and zero otherwise. Macro-averaged DOA (macro-DOA) is the average of all DOA$_i$ and micro-averaged DOA (micro-DOA) is the ratio between the aggregate number of item pairs in the correct order and the total number of item pairs checked (for further details see \cite{Fouss:2007:RCS:1263132.1263335,Freno:2009:SPE:1557019.1557059}).
		\item[Normalized Distance-based Performance Measure] The NDPM of ranking lists $\boldsymbol{r^i}$, $\boldsymbol{\pi^i}$ is defined to be:
		\begin{equation}
		\mathrm{NDPM} \triangleq \frac{D+0.5X}{N-T_r}
		\end{equation}
		The NDPM measure gives a perfect score of 0 to RS that correctly predict every preference relation asserted by the reference. The worst score of 1 is assigned to recommendation vectors that contradict every preference relation in $\boldsymbol{r^i}$ \cite{DBLP:reference/rsh/ShaniG11,DBLP:journals/jasis/Yao95}.
	\end{description}

High scores on the first three metrics ($\rho$, $\tau$, DOA) and low score on the last (NDPM), suggest that the two ranking lists \cite{DBLP:reference/rsh/ShaniG11} are ``close'', which means that the new users are more likely to receive recommendations closer to their tastes as described by their full set of ratings.

\begin{figure*}[!htpb]
	\centering
	\epsfig{file=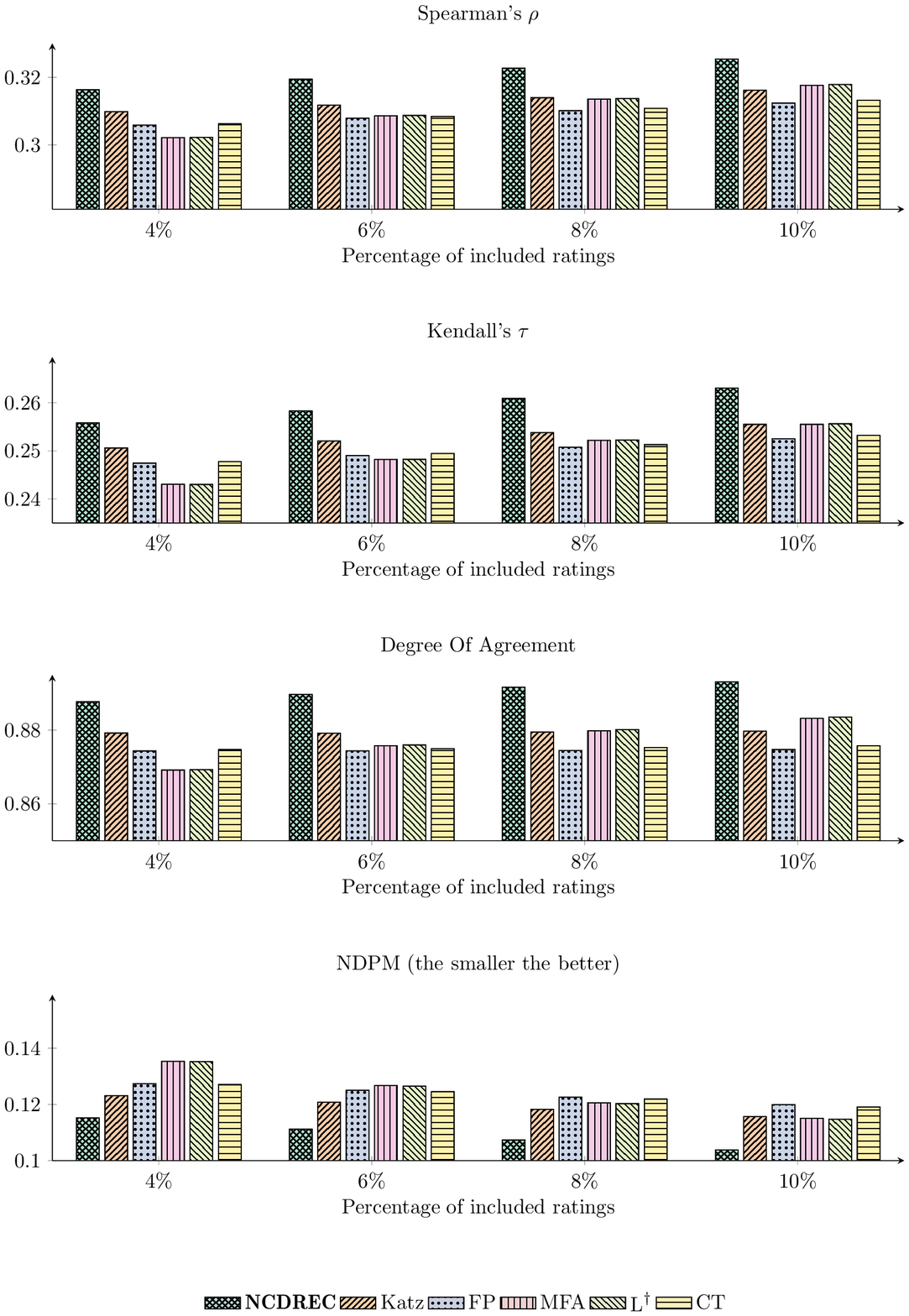,height=\textheight}
	\caption{Recommendation performance for \textit{New Users} problem}
	\label{fig:NAU_Tests}
\end{figure*}

In Figure \ref{fig:NAU_Tests} we report the average scores on all four metrics, for the set of newly added users. We see that NCDREC clearly outperforms every other method considered, achieving good results even when only 4\% of each user's ratings were included. MFA and $L\dagger$ also do well, especially as the number of ratings increases. These results are in accordance with the intuition behind our approach and the theoretical properties of the \textsc{ColdStart} subcomponent. We see that, even though new users' tastes are not yet clear, the exploitation of NCD proximity captured by matrix $\mathbf{D}$, manages to ``propagate'' this scarce rating information to the many related elements of the item space, giving our method an advantage in uncovering new users' preferences. This leads to a recommendation vector exhibiting lower sensitivity to sparsity.

\section{Conclusions and Future Work}
\label{Sec_Conclusion}
In this work we proposed NCDREC; a novel method that builds on the intuition behind \textit{Decomposability} to provide an elegant and computationally efficient framework for generating recommendations. NCDREC exploits the innately hierarchical structure of the item space, introducing the notion of \textit{NCD proximity}, which characterizes inter-level relations between the elements of the system and gives our model useful  antisparsity theoretical properties.

One very interesting direction we are currently pursuing involves  the generalization of the \textsc{ColdStart} subcomponent exploiting the functional rankings family~\cite{DampingFunctions}. In particular, based on a recently proposed, multidamping reformulation of these rankings~\cite{Gallopoulos2,Gallopoulos} that allows intuitive and fruitful interpretations of the damping functions in terms random surfing habits, one could try to capture the actual newly emerging users' behavior as they begin to explore the recommender system, and map it to suitable collections of personalized damping factors that could lead to even better recommendations. Another interesting research path that remains to be explored involves the introduction of more than one decompositions based on different criteria, and the effect it has to the theoretical properties of the \textsc{ColdStart} subcomponent. Notice, that in NCDREC this generalization can be achieved readily, through the introduction of new low-rank proximity matrices, $\mathbf{D_1,W_1, D_2, W_2, \dots}$ and associated parameters, with no effect on the dimensionality of the model.

In this work, we considered the single decomposition case. Our experiments on the \texttt{MovieLens} and the \texttt{Yahoo!R2Music} datasets, indicate that NCDREC outperforms several -- known for their antisparsity properties -- state-of-the-art graph-based algorithms in widely used performance metrics, being at the same time by far the most economical one. Note here that the random-walk approaches, FP and  CT, require to handle a graph of $(n+m+K)$ nodes and to compute $2nm$ first passage time scores. Similarly, $\mathrm{L}\dagger$, Katz and MFA, involve the inversions of an $(n+m+K)$-dimensional square matrix. In fact, only NCDREC involves matrices whose dimensions depend solely on the cardinality of the itemspace, which in most realistic applications increases slowly.

In conclusion, our findings suggest that NCDREC carries the potential of handling sparsity effectively, and produce high quality results in standard, long-tail as well as cold-start recommendation scenarios. 


\appendix

\section{Theoretical Discussion of NCDREC's Main Component}
\label{Ap:MainComponentDiscussion}

Let us consider the singular value decomposition of matrix $\mathbf{G}$,
\begin{equation}
\mathbf{G}  =  \mathbf{U}\mathbf{\Sigma}\mathbf{V}^\intercal 
\end{equation}
Multiplying from the right with $\mathbf{V}$ and using the fact that its columns denote an orthonormal set of vectors we get
\begin{equation}
\mathbf{G}\mathbf{V} =  \mathbf{U}\mathbf{\Sigma} 
\end{equation} 
Multiplying from the right with the diagonal matrix $\operatorname{Diag}(\underbrace{1,\dots,1}_{f},0,\dots,0)$ gives
\begin{eqnarray}
\mathbf{G}\begin{bmatrix}
\mathbf{V_f} & \boldsymbol{0}
\end{bmatrix} & = & \mathbf{U}\begin{bmatrix}
\mathbf{\Sigma_f} & \boldsymbol{0} \\
\boldsymbol{0} & \boldsymbol{0}
\end{bmatrix} 
\end{eqnarray}
and finally, discarding the zero columns we get 
\begin{equation}
\mathbf{G}\mathbf{V_f} =  \mathbf{U_f}\mathbf{\Sigma_f} 
\end{equation}

Now plugging this in Eq.~(\ref{Eq:MainComponentRecVectors}) we see that the recommendation vector for the user $u_i$, $\boldsymbol{\pi}_i^\intercal$ is given by:
\begin{equation}
\boldsymbol{\pi}_i^\intercal =  \mathbf{g}_{u_i}^\intercal\mathbf{V_f}\mathbf{V_f}^\intercal
\end{equation}

Notice that $\mathbf{V_f}$ contains the orthogonal set of eigenvectors of the $m\times m$ symmetric positive semidefinite matrix 
\begin{eqnarray}
\mathbf{G^\intercal G} & = & (\mathbf{R} + \epsilon\mathbf{W})^\intercal(\mathbf{R} + \epsilon\mathbf{W}) \nonumber\\
& = & (\mathbf{R}^\intercal + \epsilon\mathbf{W}^\intercal)(\mathbf{R} + \epsilon\mathbf{W})  \nonumber\\
& = & \mathbf{R}^\intercal\mathbf{R} + \epsilon(\mathbf{R}^\intercal\mathbf{W}+\mathbf{W}^\intercal\mathbf{R}) +\epsilon^2\mathbf{W}^\intercal\mathbf{W}\nonumber\\
\label{Eq:EigenRepl}
\end{eqnarray}

Thus the recommendation vectors produced by the main component of NCDREC can be seen as arising from a low dimensional eigenspace of the NCD - perturbed  inter-item similarity matrix, as seen in Eq.~(\ref{Eq:EigenRepl}).

\section{Proof of Lemma 4}
\label{Ap:LemmaProof}
\addtocounter{lemma}{-1}
\begin{lemma}
	If $\mathbf{A}$ is the transition matrix of an irreducible and aperiodic Markov chain with finite state space, and $\mathbf{B}$ the transition matrix of any Markov chain defined onto the same state space, then matrix $\mathbf{C} = \kappa\mathbf{A} + \lambda\mathbf{B}$, where $\kappa,\lambda > 0$ such that $\kappa+\lambda = 1$ denotes the transition matrix of an irreducible and aperiodic Markov chain also.
\end{lemma}

\begin{proof}
	It is easy to see that for $\kappa,\lambda > 0$ such that $\kappa+\lambda = 1$ matrix $\mathbf{C}$ is also a valid transition probability matrix. Furthermore, when $\mathbf{A}$ is irreducible there exists a positive probability path between any two given states of the corresponding Markov chain. The same path will also be valid for the Markov chain that corresponds to matrix $\mathbf{C}$, as long as $\kappa>0$. The same thing is true for the aperiodicity property, since the addition of the stochastic matrix $\mathbf{B}$ does nothing to the length of the possible paths that allow a return to any given state of the Markov chain that corresponds to matrix $\mathbf{A}$. Thus, the irreducibility and the aperiodicity of $\mathbf{A}$, together with the requirement $\kappa>0$, imply the existence of those properties to the final matrix $\mathbf{C}$, as needed. 
\end{proof}

\end{document}